\documentclass[11pt]{article}

\usepackage[letterpaper,margin=1in]{geometry}

\usepackage{amsmath,amssymb,amsthm,mathtools}
\usepackage{graphicx}
\usepackage{xcolor}
\usepackage{microtype}
\usepackage{comment}

\usepackage[colorlinks=true,linkcolor=blue,citecolor=blue,urlcolor=blue]{hyperref}
\usepackage[nameinlink,capitalise,noabbrev]{cleveref}

\usepackage{enumitem}
\usepackage{booktabs}
\usepackage{algorithm}
\usepackage{algorithmic}

\newtheorem{theorem}{Theorem}[section]
\newtheorem{lemma}[theorem]{Lemma}
\newtheorem{proposition}[theorem]{Proposition}
\newtheorem{corollary}[theorem]{Corollary}

\theoremstyle{definition}

\theoremstyle{remark}
\newtheorem{remark}[theorem]{Remark}

\newcommand{\F}{\mathbb{F}}

\newcommand{\BW}{\mathrm{BW}}
\newcommand{\IO}{\mathrm{IO}}
\newcommand{\PG}{\mathrm{PG}}

\newcommand{\avg}{\mathrm{avg}}
\newcommand{\rank}{\mathrm{rank}}
\newcommand{\nz}{\mathrm{nz}}

\newcommand{\col}{\mathrm{col}}

\begin{document}

\begin{center}
  {\LARGE\bfseries
  The Incidence-Multiplicity Bound for Linear\\[0.3em]
  Exact Repair in MDS Array Codes\par}
\end{center}

\vspace{0.6em}
\begin{center}
  {\large
    Huawei Wu\textsuperscript{1,*}
  }\\[0.55em]
  {\normalsize
    \textsuperscript{1}Shanghai Fintelli Box Technology Co., Ltd., Shanghai, 200127, China\\[0.35em]
    \textsuperscript{*}Corresponding author.
    \href{mailto:wuhuawei1996@gmail.com}{\texttt{wuhuawei1996@gmail.com}}
  }
\end{center}

\begin{abstract}
  We study linear exact repair for $(n,k,\ell)$ MDS array codes over $\mathbb{F}_q$, with redundancy $r=n-k$, in the regime where $q$, $r$, and $\ell$ are fixed and the code length $n$ varies.
  
  A recent projective counting argument gives a general lower bound on repair bandwidth and repair I/O in this setting. While this bound is attained over a broad interval of code lengths in the two-parity case, it is not attained once $r\ge 3$ and $\ell\ge 2$.
  
  In this paper, we refine the counting argument behind this bound and establish a sharper lower bound, which we call the incidence-multiplicity bound. We prove that for every $(n,k,\ell)$ MDS array code over $\mathbb{F}_q$ with $r\ge 2$, both the average and worst-case repair bandwidth, as well as the average and worst-case repair I/O, are at least
  $$\ell(n-1)-(r-1)\frac{q^\ell-1}{q-1}.$$
  This bound agrees with the earlier projective counting bound when $r=2$, and is strictly stronger for every $r\ge 3$.
  
  We also show that the incidence-multiplicity bound is sharp in a broad parameter range. Assume that $\ell\ge 2$, $r\ge 2$, $(r-1)\mid(q-1)$ and $(q-1)/(r-1)\ge 2$. Then for every integer $n$ satisfying
  $$2(r-1)\frac{q^\ell-1}{q-1}\le n\le q^\ell+1,$$
  there exists an $(n,n-r,\ell)$ MDS array code over $\F_q$ that attains the incidence-multiplicity bound simultaneously for both repair bandwidth and repair I/O. These codes arise from field reduction of a normal rational curve.
  
  Together, these results reveal incidence multiplicity as the governing geometric principle for linear exact repair in MDS array codes beyond the two-parity case.
  \end{abstract}

%
\newpage
\tableofcontents
\newpage

\section{Introduction}

Erasure coding is widely used in distributed storage systems to provide reliability against
node failures. Among erasure codes, MDS codes are especially attractive because they
achieve the optimal redundancy--reliability trade-off. In such a system, one frequently
faces the problem of repairing a failed node from the data stored in the surviving nodes.
A central measure of the efficiency of this process is the \emph{repair bandwidth}, namely
the total amount of information downloaded from the helper nodes. This naturally leads to
the question of how to design MDS codes that minimize repair bandwidth.

Compared with scalar MDS codes, \emph{MDS array codes} allow a finer-grained organization of the
data stored in each node, which can substantially reduce repair bandwidth. This has made
the \emph{sub-packetization} level, that is, the number of smaller units into which the
contents of each node are split, a central structural parameter in the study of repair
efficiency~\cite{ramkumar2022codes}.

For single-node repair, the regenerating-code framework provides an information-theoretic
lower bound, namely the \emph{cut-set bound}. At the \emph{minimum-storage point} of the resulting
trade-off, one obtains \emph{minimum-storage regenerating (MSR) codes}, which remain MDS while
achieving optimal repair bandwidth~\cite{dimakis2010network}. The obstacle is that this optimum is expensive to realize: in the high-rate regime,
exact-repair MSR codes provably require exponential, or at least very large, sub-packetization~\cite{balaji2018tight,alrabiah2019exponential,balaji2022lower}. From a practical perspective, such excessively fine partitioning is also undesirable,
since it may lead to fragmented and often non-contiguous disk accesses~\cite{gan2025revisiting,shen2025survey}. This has led to a broad line of work on \emph{MDS array codes} with small, or even
constant, sub-packetization, with the aim of determining the best achievable repair bandwidth without insisting on the MSR point~\cite{rashmi2017piggybacking,kralevska2018hashtag,rawat2017msr,kralevska2018explicit,ramkumar2022codes,ramkumar2025varepsilon}. A central open direction is to understand the trade-off between repair bandwidth, sub-packetization, and field size for MDS array codes~\cite[Open Problem~9]{ramkumar2022codes}.

Besides repair bandwidth, another important measure of repair efficiency is the
\emph{repair I/O}, namely the total amount of information accessed at the helper nodes
during repair. This quantity has attracted increasing attention in recent
years~\cite{dau2018repair,li2019costs,liu2024formula,liu2025calculating}. Accordingly, in
designing codes for efficient repair, one would ideally like to simultaneously minimize both repair
bandwidth and repair I/O.

In this paper, we consider \emph{linear exact repair} for $(n,k,\ell)$ MDS array codes over the finite field $\F_q$ of size $q$, with redundancy $r=n-k$, in the regime where the field size $q$, the redundancy $r$, and the sub-packetization level $\ell$ are fixed while the code length $n$ varies. Our focus is on lower bounds for repair bandwidth and repair I/O, and on
the structural mechanisms underlying those bounds.

A recent paper \cite{liu2026linear} introduced an intrinsic subspace formulation of linear
exact repair and, from that viewpoint, proved a general lower bound by a projective counting
argument. For an $(n,k,\ell)$ MDS array code $\mathcal{C}$ over $\F_q$ with redundancy $r=n-k\ge 2$,
let $\beta_{\avg}(\mathcal{C})$ and $\beta_{\max}(\mathcal{C})$ denote the average and worst-case repair
bandwidth, and let $\gamma_{\avg}(\mathcal{C})$ and $\gamma_{\max}(\mathcal{C})$ denote the corresponding
repair I/O parameters. That paper proves that
\begin{equation}\label{eq:projective-counting-bound}
\beta_{\avg}(\mathcal{C}),\ \beta_{\max}(\mathcal{C}),\ \gamma_{\avg}(\mathcal{C}),\ \gamma_{\max}(\mathcal{C})
\;\ge\;
\ell(n-1)-\frac{q^{(r-1)\ell}-1}{q-1}.
\end{equation}
It also shows a sharp contrast between the regimes $r=2$ and $r\ge 3$: in the two-parity
case, constructions from finite geometry attain the bound over a broad interval of code
lengths simultaneously for all four quantities, whereas once $r\ge 3$ and $\ell\ge 2$ it
is never attained. Thus the projective counting bound provides a useful general
obstruction, but beyond the two-parity case it does not capture the correct lower-bound scale for repair bandwidth and repair I/O.

The starting point of the present paper is that the projective counting argument of
\cite{liu2026linear} uses only a relatively coarse packing phenomenon. More precisely, it
counts projective pieces inside a repair subspace by exploiting pairwise disjointness.
What it does not fully use is the stronger $r$-wise spanning structure imposed by the MDS
condition. We show that once this higher-order structure is brought into the counting
problem, one obtains a strictly sharper lower bound. We call the resulting estimate the
\emph{incidence-multiplicity bound}.

Our first main result states that for every $(n,k,\ell)$ MDS array code $\mathcal{C}$ over $\F_q$ with redundancy $r=n-k\ge 2$, one has
\[
\beta_{\avg}(\mathcal{C}),\ \beta_{\max}(\mathcal{C}),\ \gamma_{\avg}(\mathcal{C}),\ \gamma_{\max}(\mathcal{C})
\;\ge\;
\ell(n-1)-(r-1)\frac{q^\ell-1}{q-1}.
\]
When $r=2$, this agrees with the projective counting bound of \cite{liu2026linear}; when
$r\ge 3$, it is strictly stronger, replacing the correction term
\[
\frac{q^{(r-1)\ell}-1}{q-1}
\]
by
\[
(r-1)\frac{q^\ell-1}{q-1}.
\]

Our second main result shows that this new bound is sharp on a broad explicit family.
Under the assumptions
\[
\ell\ge 2,\qquad r\ge 2,\qquad (r-1)\mid(q-1),\qquad \frac{q-1}{r-1}\ge 2,
\]
we construct, for every integer $n$ satisfying
\[
2(r-1)\frac{q^\ell-1}{q-1}\le n\le q^\ell+1,
\]
an $(n,n-r,\ell)$ MDS array code over $\F_q$ such that the incidence-multiplicity bound is
attained simultaneously by the average and worst-case repair bandwidth and by the average
and worst-case repair I/O. These constructions arise from field reduction of a normal
rational curve.

Taken together, these results show that beyond the two-parity case, incidence multiplicity
rather than projective packing is the geometric principle that governs linear exact repair
in MDS array codes.

The rest of the paper is organized as follows. In Section~2 we briefly recall the standard
matrix formulation of linear exact repair and the intrinsic subspace reformulation proposed in \cite{liu2026linear}. In Section~3 we prove the incidence-multiplicity bound. In Section~4 we construct a broad family of MDS array codes, arising from field reduction of a normal rational curve, for which the new bound is attained simultaneously by the
average and worst-case repair bandwidth and by the average and worst-case repair I/O. We conclude in Section~5 with
some further remarks and open questions.

\section{Preliminaries and Intrinsic Setup}\label{sec:prelim}

In this section we briefly recall the standard matrix formulation of linear exact repair and
the intrinsic subspace reformulation introduced in \cite{liu2026linear}. We only record the
notation and consequences needed in the sequel.

\subsection{Linear Exact Repair in Matrix Form}\label{subsec:matrix}

Throughout the paper, for a positive integer $m$, we write $[m]:=\{1,2,\dots,m\}$.

Fix integers $n,k,\ell\in\mathbb{N}^+$ with $1\le k<n$, and write
\[
r:=n-k.
\]
An $(n,k,\ell)$ \emph{MDS array code} over $\F_q$ is an $\F_q$-linear subspace
\[
\mathcal{C}\le (\F_q^\ell)^n
\]
of dimension $k\ell$. Its elements are written as
\[
C=(C_1,\dots,C_n),\qquad C_i\in\F_q^\ell,
\]
where $C_i$ is the block stored in node $i$.

A convenient description of $\mathcal{C}$ is via a block parity-check matrix
\[
H=[\,H_1\ H_2\ \cdots\ H_n\,]\in\F_q^{r\ell\times n\ell},
\qquad
H_i\in\F_q^{r\ell\times \ell},
\]
of full row rank $r\ell$, so that
\[
\mathcal{C}
=
\ker(H)
=
\left\{
(C_1,\dots,C_n)\in(\F_q^\ell)^n:\ \sum_{i=1}^n H_iC_i=0
\right\}.
\]
The code $\mathcal{C}$ is MDS if for every subset $I\subseteq[n]$ with $|I|=r$, the square
block matrix
\[
H_I:=[\,H_i\,]_{i\in I}\in\F_q^{r\ell\times r\ell}
\]
is invertible.

Suppose that node $i\in[n]$ fails. A \emph{linear repair scheme} for node $i$ is specified
by a matrix
\[
M\in\F_q^{\ell\times r\ell}.
\]
Multiplying the parity-check equations by $M$ gives
\[
\sum_{j=1}^n (MH_j)C_j=0.
\]
If $MH_i$ is invertible, then the missing block $C_i$ is determined uniquely by
\[
C_i=-(MH_i)^{-1}\sum_{j\ne i}(MH_j)C_j.
\]
Accordingly, whenever $MH_i$ is invertible, we say that $M$ \emph{repairs} node $i$.

For such a repair matrix $M$, define the repair bandwidth
\[
\mathrm{BW}_i(M):=\sum_{j\ne i}\rank(MH_j).
\]
This is the total number of $\F_q$-symbols downloaded from the helper nodes.

To measure the amount of accessed information, let $\nz(A)$ denote the number of nonzero
columns of a matrix $A$. Since $(MH_j)C_j$ depends only on those coordinates of $C_j$
corresponding to nonzero columns of $MH_j$, define the repair I/O by
\[
\mathrm{IO}_i(M):=\sum_{j\ne i}\nz(MH_j).
\]
Clearly,
\[
\mathrm{IO}_i(M)\ge \mathrm{BW}_i(M)
\]
for every repair matrix $M$.

For each node $i\in[n]$, let
\[
\mathcal{M}_i:=\{\,M\in\F_q^{\ell\times r\ell}: MH_i \text{ is invertible}\,\}
\]
be the set of repair matrices for node $i$. We then define the optimal per-node repair
bandwidth and repair I/O by
\[
\beta_i(\mathcal{C}):=\min_{M\in\mathcal{M}_i}\mathrm{BW}_i(M),
\qquad
\gamma_i(\mathcal{C}):=\min_{M\in\mathcal{M}_i}\mathrm{IO}_i(M).
\]
Aggregating over all nodes gives
\[
\beta_{\avg}(\mathcal{C}):=\frac1n\sum_{i=1}^n\beta_i(\mathcal{C}),
\qquad
\beta_{\max}(\mathcal{C}):=\max_{i\in[n]}\beta_i(\mathcal{C}),
\]
and
\[
\gamma_{\avg}(\mathcal{C}):=\frac1n\sum_{i=1}^n\gamma_i(\mathcal{C}),
\qquad
\gamma_{\max}(\mathcal{C}):=\max_{i\in[n]}\gamma_i(\mathcal{C}).
\]

\subsection{Intrinsic Subspace Reformulation}\label{subsec:intrinsic}

We now pass to the intrinsic reformulation from \cite{liu2026linear}. Set
\[
  \mathbb{V}:=\F_q^{r\ell}.
\]
Write each parity-check block as
\[
  H_i=[\,h_{i,1}\ \cdots\ h_{i,\ell}\,],
  \qquad h_{i,t}\in\mathbb{V},
\]
and define
\[
  \mathcal{H}_i:=\col(H_i)\le \mathbb{V}.
\]
Since $\mathcal{C}$ is MDS, each $\mathcal{H}_i$ has dimension $\ell$, and the family
$\mathcal{H}_1,\dots,\mathcal{H}_n$ satisfies
\[
  \sum_{j\in J}\mathcal{H}_j=\mathbb{V}
  \qquad\text{for every }J\subseteq[n]\text{ with }|J|=r.
\]

For a fixed node $i\in[n]$, define
\[
  \mathcal{W}_i
  :=
  \{\,W\le \mathbb{V}:\dim(W)=(r-1)\ell,\ W\cap \mathcal{H}_i=\{0\}\,\}.
\]
By \cite{liu2026linear}, if $M\in\mathcal{M}_i$ and $W=\ker(M)$, then
\[
  W\in\mathcal{W}_i
\]
and, for every $j\ne i$,
\[
  \rank(MH_j)=\ell-\dim(W\cap \mathcal{H}_j).
\]
Hence
\[
  \BW_i(M)=\ell(n-1)-\sum_{j\ne i}\dim(W\cap \mathcal{H}_j).
\]
This motivates
\[
  \alpha_i:=\max_{W\in\mathcal{W}_i}\sum_{j\ne i}\dim(W\cap \mathcal{H}_j),
\]
so that
\[
  \beta_i(\mathcal{C})=\ell(n-1)-\alpha_i.
\]
With
\[
  \alpha_{\avg}:=\frac1n\sum_{i=1}^n\alpha_i,
  \qquad
  \alpha_{\min}:=\min_{i\in[n]}\alpha_i,
\]
we obtain
\[
  \beta_{\avg}(\mathcal{C})=\ell(n-1)-\alpha_{\avg},
  \qquad
  \beta_{\max}(\mathcal{C})=\ell(n-1)-\alpha_{\min}.
\]

To treat repair I/O, define the projective column set
\[
  X_i:=\{\langle h_{i,1}\rangle,\dots,\langle h_{i,\ell}\rangle\}\subseteq\mathbb{P}(\mathcal{H}_i).
\]
For $W\in\mathcal{W}_i$ and $j\ne i$, let
\[
  z_j(W):=\bigl|\{\,t\in[\ell]: h_{j,t}\in W\,\}\bigr|
  =|X_j\cap\mathbb{P}(W)|.
\]
Again by \cite{liu2026linear},
\[
  \IO_i(M)=\ell(n-1)-\sum_{j\ne i}z_j(W).
\]
This motivates
\[
  \lambda_i:=\max_{W\in\mathcal{W}_i}\sum_{j\ne i}z_j(W),
\]
and therefore
\[
  \gamma_i(\mathcal{C})=\ell(n-1)-\lambda_i.
\]
With
\[
  \lambda_{\avg}:=\frac1n\sum_{i=1}^n\lambda_i,
  \qquad
  \lambda_{\min}:=\min_{i\in[n]}\lambda_i,
\]
we obtain
\[
  \gamma_{\avg}(\mathcal{C})=\ell(n-1)-\lambda_{\avg},
  \qquad
  \gamma_{\max}(\mathcal{C})=\ell(n-1)-\lambda_{\min}.
\]

Finally, we record the realization facts that will be used later. Any family of
$\ell$-dimensional subspaces
\[
  \mathcal{H}_1,\dots,\mathcal{H}_n\le\mathbb{V}
\]
satisfying
\[
  \sum_{j\in J}\mathcal{H}_j=\mathbb{V}
  \qquad\text{for every }J\subseteq[n]\text{ with }|J|=r
\]
gives rise to an $(n,n-r,\ell)$ MDS array code over $\F_q$. Moreover, if for each $i\in[n]$
one is given a set
\[
  X_i\subseteq\mathbb{P}(\mathcal{H}_i)
\]
consisting of $\ell$ distinct points spanning $\mathcal{H}_i$, then one obtains a
parity-check realization whose induced projective column set at node $i$ is exactly $X_i$.

From this point on, we work entirely in the intrinsic language of the node subspaces
$\mathcal{H}_i$, the feasible repair subspaces $\mathcal{W}_i$, and the projective column
sets $X_i$.

\section{The Incidence-Multiplicity Bound}\label{sec:imbound}

The projective counting bound of \cite{liu2026linear} is sharp in the two-parity case, but
for $r\ge 3$ and $\ell\ge 2$ its correction term is too large to be compatible with the
general MDS length bound. Indeed, by \cite[Remark~3.2]{liu2026linear}, equality in the projective counting bound
would force
\[
n\ge 1+\frac{q^{(r-1)\ell}-1}{q-1},
\]
whereas every $(n,n-r,\ell)$ MDS array code over $\F_q$ satisfies
\[
n\le q^\ell+r-1
\]
by \cite[Lemma~3.3]{liu2026linear}. When $r\ge 3$ and $\ell\ge 2$, these two inequalities are incompatible, because the lower
bound involves the term $q^{(r-1)\ell}$, whereas the general MDS length bound grows only on the order of $q^\ell$. This order mismatch suggests that the projective packing argument counts
at the wrong scale beyond the two-parity case, and motivates a sharper counting
principle.

The next proposition refines the projective counting argument of \cite{liu2026linear}.
Instead of counting projective points inside a repair subspace, we pass to the quotient by
that subspace and count incidences in the dual projective space.

\begin{proposition}\label{prop:dual-counting}
Let $\mathcal{H}_1,\dots,\mathcal{H}_n\le \mathbb{V}$ be $\ell$-dimensional subspaces over
$\F_q$ such that
\[
\sum_{j\in J}\mathcal{H}_j=\mathbb{V}
\qquad\text{for every }J\subseteq[n]\text{ with }|J|=r.
\]
Fix $i\in[n]$, and let $W\le \mathbb{V}$ be any subspace of codimension $\ell$, i.e., $\dim_{\F_q}W=r\ell-\ell$.
 For each $j\ne i$, set
\[
t_j:=\dim(W\cap \mathcal{H}_j).
\]
Then
\[
\sum_{j\ne i}\frac{q^{t_j}-1}{q-1}
\;\le\;
(r-1)\frac{q^\ell-1}{q-1}.
\]
In particular,
\[
\sum_{j\ne i} t_j
\;\le\;
(r-1)\frac{q^\ell-1}{q-1}.
\]
\end{proposition}

\begin{proof}
  Set
  \[
  Q:=\mathbb{V}/W,
  \]
  so that $\dim_{\F_q}Q=\ell$, and let $\pi:\mathbb{V}\to Q$ be the quotient map. For each $j\ne i$, define
\[
K_j:=\pi(\mathcal{H}_j)=(\mathcal{H}_j+W)/W\le Q.
\]
Since $\ker(\pi|_{\mathcal{H}_j})=\mathcal{H}_j\cap W$,
 we have
\[
\dim K_j
=
\dim \mathcal{H}_j-\dim(\mathcal{H}_j\cap W)
=
\ell-t_j.
\]

Now pass to the dual space $Q^*$. For each $j\ne i$, define
\[
U_j:=K_j^\perp\le Q^*,
\]
where
\[
K_j^\perp:=\{\,\varphi\in Q^*:\ \varphi(x)=0\text{ for all }x\in K_j\,\}.
\]
Then
\[
\dim U_j
=
\dim Q-\dim K_j
=
\ell-(\ell-t_j)
=
t_j.
\]

We claim that for every subset $J\subseteq[n]\setminus\{i\}$ with $|J|=r$,
\[
\bigcap_{j\in J}U_j=\{0\}.
\]
Indeed, by the $r$-wise spanning assumption,
\[
\sum_{j\in J}\mathcal{H}_j=\mathbb{V}.
\]
Applying $\pi$, we obtain
\[
\sum_{j\in J}K_j
=
\sum_{j\in J}\pi(\mathcal{H}_j)
=
\pi\!\left(\sum_{j\in J}\mathcal{H}_j\right)
=
Q.
\]
Taking annihilators in $Q^*$ gives
\[
\bigcap_{j\in J}U_j
=
\bigcap_{j\in J}K_j^\perp
=
\left(\sum_{j\in J}K_j\right)^\perp
=
Q^\perp
=
\{0\}.
\]

Therefore, in the projective space $\mathbb{P}(Q^*)$, each projective point lies in at
most $r-1$ of the projective subspaces $\mathbb{P}(U_j)$, $j\ne i$; otherwise, if some
point lay in $\mathbb{P}(U_j)$ for $r$ distinct indices, then a nonzero vector
representing that point would belong to the intersection of those $r$ corresponding
subspaces, contradicting the claim.

Counting incidences between projective points of $\mathbb{P}(Q^*)$ and the family
$\{\mathbb{P}(U_j):j\ne i\}$, we obtain
\[
\sum_{j\ne i} |\mathbb{P}(U_j)|
\;\le\;
(r-1)|\mathbb{P}(Q^*)|.
\]
Since
\[
|\mathbb{P}(U_j)|=\frac{q^{\dim U_j}-1}{q-1}
=
\frac{q^{t_j}-1}{q-1},
\qquad
|\mathbb{P}(Q^*)|=\frac{q^\ell-1}{q-1},
\]
it follows that
\[
\sum_{j\ne i}\frac{q^{t_j}-1}{q-1}
\;\le\;
(r-1)\frac{q^\ell-1}{q-1}.
\]

Finally, for every integer $t\ge 0$,
\[
t \le \frac{q^t-1}{q-1},
\]
hence
\[
\sum_{j\ne i} t_j
\;\le\;
\sum_{j\ne i}\frac{q^{t_j}-1}{q-1}
\;\le\;
(r-1)\frac{q^\ell-1}{q-1}.
\]
This completes the proof.
\end{proof}

\begin{remark}\label{rem:multiplicity-equality}
  Keep the notation of the proof of Proposition~\ref{prop:dual-counting}. If
  \[
  \sum_{j\ne i} t_j=(r-1)\frac{q^\ell-1}{q-1},
  \]
  then both inequalities
  \[
  \sum_{j\ne i} t_j
  \;\le\;
  \sum_{j\ne i}\frac{q^{t_j}-1}{q-1}
  \;\le\;
  (r-1)\frac{q^\ell-1}{q-1}
  \]
  must be equalities. Consequently:
  \begin{enumerate}[label=(\arabic*)]
  \item $t_j\in\{0,1\}$ for every $j\ne i$;
  \item every point of $\mathbb{P}(Q^*)$ lies in exactly $r-1$ of the subspaces
  $\mathbb{P}(U_j)$, $j\ne i$;
  \item
  \[
  \#\{\,j\ne i:\dim(W\cap \mathcal{H}_j)=1\,\}
  =
  (r-1)\frac{q^\ell-1}{q-1},
  \]
  and hence equality can hold only if
  \begin{equation}\label{eq:necessary-length-for-equality}
  n\ge 1+(r-1)\frac{q^\ell-1}{q-1}.
  \end{equation}
  \end{enumerate}
  Equivalently, in the equality case, the nonzero subspaces $U_j$ form a multiset in which
  every point of $\mathbb{P}(Q^*)$ occurs with multiplicity exactly $r-1$.
  \end{remark}

The proposition immediately yields the new lower bound for repair bandwidth and repair I/O.

\begin{theorem}[Incidence-multiplicity bound]\label{thm:imbound}
  Let $\mathcal{C}$ be an $(n,k,\ell)$ MDS array code over $\F_q$ with redundancy
  $r=n-k\ge 2$. Then for every $i\in[n]$,
  \[
  \beta_i(\mathcal{C}),\ \gamma_i(\mathcal{C})
  \;\ge\;
  \ell(n-1)-(r-1)\frac{q^\ell-1}{q-1}.
  \]
  Hence
  \[
  \beta_{\avg}(\mathcal{C}),\ \beta_{\max}(\mathcal{C}),\ 
  \gamma_{\avg}(\mathcal{C}),\ \gamma_{\max}(\mathcal{C})
  \;\ge\;
  \ell(n-1)-(r-1)\frac{q^\ell-1}{q-1}.
  \]
  \end{theorem}
  
  \begin{proof}
  Fix $i\in[n]$. For each $W\in\mathcal{W}_i$, Proposition~\ref{prop:dual-counting} gives
  \[
  \sum_{j\ne i}\dim(W\cap \mathcal{H}_j)\le (r-1)\frac{q^\ell-1}{q-1}.
  \]
  Taking the maximum over $W\in\mathcal{W}_i$, we obtain
  \[
  \alpha_i\le (r-1)\frac{q^\ell-1}{q-1}.
  \]
  Hence
  \[
  \beta_i(\mathcal{C})
  =
  \ell(n-1)-\alpha_i
  \;\ge\;
  \ell(n-1)-(r-1)\frac{q^\ell-1}{q-1}.
  \]
  Since $\gamma_i(\mathcal{C})\ge \beta_i(\mathcal{C})$, the same lower bound also holds for
  $\gamma_i(\mathcal{C})$. The bounds for $\beta_{\avg}(\mathcal{C})$, $\beta_{\max}(\mathcal{C})$,
  $\gamma_{\avg}(\mathcal{C})$, and $\gamma_{\max}(\mathcal{C})$ follow immediately.
  \end{proof}

  \begin{remark}
    When $r\ge 3$ and $\ell\ge 2$, Theorem~\ref{thm:imbound} is strictly stronger than the projective counting
    bound of \cite[Theorem~1.1]{liu2026linear}; when $r=2$, the two bounds coincide.
    \end{remark}

    We have the following necessary condition for attaining the incidence-multiplicity bound.

    \begin{corollary}\label{cor:r-le-q}
      Assume that $\ell\ge 2$. Let $\mathcal{C}$ be an $(n,k,\ell)$ MDS array code over $\F_q$
      with redundancy $r=n-k$. If the incidence-multiplicity bound is attained by at least one
      of
      \[
      \beta_i(\mathcal{C}),\ \gamma_i(\mathcal{C})\quad (i\in[n]),
      \qquad
      \beta_{\avg}(\mathcal{C}),\ \beta_{\max}(\mathcal{C}),\ 
      \gamma_{\avg}(\mathcal{C}),\ \gamma_{\max}(\mathcal{C}),
      \]
      then necessarily
      \[
      r\le q.
      \]
      \end{corollary}

      \begin{proof}
        If the incidence-multiplicity bound is attained by one of the quantities listed above, then by Inequality~\eqref{eq:necessary-length-for-equality}, one has
        \[
        n\ge 1+(r-1)\frac{q^\ell-1}{q-1}.
        \]
        On the other hand, \cite[Lemma~3.3]{liu2026linear} gives
        \[
        n\le q^\ell+r-1.
        \]
        Hence
        \[
        (r-1)\left(\frac{q^\ell-1}{q-1}-1\right)\le q^\ell-1.
        \]
        Since $\ell\ge 2$,
        \[
        \frac{q^\ell-1}{q-1}-1=q+q^2+\cdots+q^{\ell-1}\ge q^{\ell-1},
        \]
        so
        \[
        (r-1)q^{\ell-1}\le q^\ell-1<q^\ell.
        \]
        Therefore \(r-1<q\), i.e. \(r\le q\).
        \end{proof}

        We close this section by noting that Proposition~\ref{prop:dual-counting} extends naturally
to higher-dimensional subspaces.
        
        \begin{proposition}\label{prop:gaussian-hierarchy}
        In the setting of Proposition~\ref{prop:dual-counting}, for every integer \(s\) with
        \(1\le s\le \ell\), one has
        \[
        \sum_{j\ne i}\binom{t_j}{s}_q
        \;\le\;
        (r-1)\binom{\ell}{s}_q,
        \]
        where \(\binom{a}{s}_q\) denotes the Gaussian binomial coefficient, defined by
\[
\binom{a}{s}_q
:=
\frac{(q^a-1)(q^a-q)\cdots(q^a-q^{s-1})}
     {(q^s-1)(q^s-q)\cdots(q^s-q^{s-1})}
\]
for \(a\ge s\), and \(\binom{a}{s}_q:=0\) for \(a<s\).
In particular, the case \(s=1\) recovers Proposition~\ref{prop:dual-counting}.
        \end{proposition}
        
        \begin{proof}
        Keep the notation of the proof of Proposition~\ref{prop:dual-counting}. Fix
        \(1\le s\le \ell\). Since the intersection of any \(r\) of the subspaces \(U_j\) is trivial,
        every \(s\)-dimensional subspace of \(Q^*\) is contained in at most \(r-1\) of the
        subspaces \(U_j\). Counting incidences between \(s\)-dimensional subspaces of \(Q^*\) and
        the family \(\{U_j:j\ne i\}\), we obtain
        \[
        \sum_{j\ne i}\#\{\,L\le U_j:\dim L=s\,\}
        \;\le\;
        (r-1)\#\{\,L\le Q^*:\dim L=s\,\}.
        \]
        Now
        \[
        \#\{\,L\le U_j:\dim L=s\,\}=\binom{t_j}{s}_q,
        \qquad
        \#\{\,L\le Q^*:\dim L=s\,\}=\binom{\ell}{s}_q,
        \]
        and the result follows.
        \end{proof}

        \section{Sharp Constructions from Field Reduction of a Normal Rational Curve}\label{sec:construction}
        We now show that the incidence-multiplicity bound is sharp in a broad parameter range.
        The construction is a higher-redundancy generalization of the two-parity finite-geometric
        construction in \cite{liu2026linear}.
        
        Write
        \[
        t_\ell(q):=\frac{q^\ell-1}{q-1}.
        \]
        Throughout this section, we work in
        \[
        \mathbb{V}:=\F_{q^\ell}^r,
        \]
        viewed as an $\F_q$-vector space of dimension $r\ell$.
        
        Consider the standard normal rational curve in $\PG(r-1,q^\ell)$, namely the image of
        \[
        \nu:\PG(1,q^\ell)\longrightarrow \PG(r-1,q^\ell),\qquad
        [s:t]\longmapsto [s^{r-1}:s^{r-2}t:\cdots:st^{r-2}:t^{r-1}].
        \]
        Its $\F_{q^\ell}$-rational points are
        \[
        \{[1:c:c^2:\cdots:c^{r-1}]:c\in\F_{q^\ell}\}\cup\{[0:\cdots:0:1]\}.
        \]
        For each $c\in\F_{q^\ell}$, define
        \[
        \mathcal{H}_c:=\{x(1,c,c^2,\dots,c^{r-1}):x\in\F_{q^\ell}\},
        \qquad
        \mathcal{H}_\infty:=\{(0,\dots,0,x):x\in\F_{q^\ell}\}.
        \]
        Then each $\mathcal{H}_c$ and $\mathcal{H}_\infty$ is an $\ell$-dimensional $\F_q$-subspace
        of $\mathbb{V}$, obtained by field reduction from the corresponding
        $\F_{q^\ell}$-rational point on the normal rational curve. When $r=2$, this specializes
        to the projective-line/Desarguesian-spread picture used in \cite{liu2026linear}.
        
        \begin{lemma}\label{lem:nrc-mds}
        The sum of any $r$ distinct subspaces among
        \[
        \{\mathcal{H}_c:c\in\F_{q^\ell}\}\cup\{\mathcal{H}_\infty\}
        \]
        is direct. In particular, they form an $(q^\ell+1,q^\ell+1-r,\ell)$ MDS
        array-code skeleton over $\F_q$.
        \end{lemma}
        
        \begin{proof}
        We consider two cases. First, let $c_1,\dots,c_r\in\F_{q^\ell}$ be distinct. Suppose
        \[
        v_1+\cdots+v_r=0,
        \qquad v_j\in\mathcal{H}_{c_j}.
        \]
        Writing
        \[
        v_j=x_j(1,c_j,c_j^2,\dots,c_j^{r-1}),
        \qquad x_j\in\F_{q^\ell},
        \]
        we obtain
        \[
        \sum_{j=1}^r x_j(1,c_j,c_j^2,\dots,c_j^{r-1})=0.
        \]
        Equivalently,
        \[
        \begin{bmatrix}
        1 & 1 & \cdots & 1 \\
        c_1 & c_2 & \cdots & c_r \\
        c_1^2 & c_2^2 & \cdots & c_r^2 \\
        \vdots & \vdots & & \vdots \\
        c_1^{r-1} & c_2^{r-1} & \cdots & c_r^{r-1}
        \end{bmatrix}
        \begin{bmatrix}
        x_1\\ \vdots\\ x_r
        \end{bmatrix}
        =0.
        \]
        Since this is a Vandermonde matrix with distinct parameters $c_1,\dots,c_r$, it is
        invertible over $\F_{q^\ell}$. Hence $x_1=\cdots=x_r=0$, so the sum is direct.
        
        Now let $c_1,\dots,c_{r-1}\in\F_{q^\ell}$ be distinct, and consider also
        $\mathcal{H}_\infty$. Suppose
        \[
        v_1+\cdots+v_{r-1}+v_\infty=0,
        \qquad v_j\in\mathcal{H}_{c_j},\quad v_\infty\in\mathcal{H}_\infty.
        \]
        Writing
        \[
        v_j=x_j(1,c_j,c_j^2,\dots,c_j^{r-1}),
        \qquad
        v_\infty=(0,\dots,0,x_\infty),
        \]
        with $x_1,\dots,x_{r-1},x_\infty\in\F_{q^\ell}$, we obtain
        \[
        \sum_{j=1}^{r-1}x_j(1,c_j,c_j^2,\dots,c_j^{r-1})+(0,\dots,0,x_\infty)=0.
        \]
        Equivalently,
        \begin{equation}\label{eq:vandermonde-system}
        \begin{bmatrix}
        1 & 1 & \cdots & 1 & 0 \\
        c_1 & c_2 & \cdots & c_{r-1} & 0 \\
        c_1^2 & c_2^2 & \cdots & c_{r-1}^2 & 0 \\
        \vdots & \vdots & & \vdots & \vdots \\
        c_1^{r-1} & c_2^{r-1} & \cdots & c_{r-1}^{r-1} & 1
        \end{bmatrix}
        \begin{bmatrix}
        x_1\\ \vdots\\ x_{r-1}\\ x_\infty
        \end{bmatrix}
        =0.
        \end{equation}
        The first $r-1$ rows already give
        \[
        \begin{bmatrix}
        1 & 1 & \cdots & 1 \\
        c_1 & c_2 & \cdots & c_{r-1} \\
        c_1^2 & c_2^2 & \cdots & c_{r-1}^2 \\
        \vdots & \vdots & & \vdots \\
        c_1^{r-2} & c_2^{r-2} & \cdots & c_{r-1}^{r-2}
        \end{bmatrix}
        \begin{bmatrix}
        x_1\\ \vdots\\ x_{r-1}
        \end{bmatrix}
        =0.
        \]
        Since this is a Vandermonde matrix with distinct parameters $c_1,\dots,c_{r-1}$, we get
        $x_1=\cdots=x_{r-1}=0$. The last row of Equation~\eqref{eq:vandermonde-system} then gives $x_\infty=0$. Hence the sum is direct.
        
        Thus the sum of any $r$ distinct subspaces among
        \[
        \{\mathcal{H}_c:c\in\F_{q^\ell}\}\cup\{\mathcal{H}_\infty\}
        \]
        is direct.
        \end{proof}
        
        Next we define the repair subspaces that will realize equality in the
        incidence-multiplicity bound. Let
        \[
        \Sigma:=\{x\in\F_{q^\ell}^\times:N_{\F_{q^\ell}/\F_q}(x)=1\}\le \F_{q^\ell}^\times.
        \]
        Then $\Sigma$ is the kernel of the norm map
        \[
        N_{\F_{q^\ell}/\F_q}:\F_{q^\ell}^\times\to\F_q^\times,
        \]
        and
        \[
        |\Sigma|=t_\ell(q)=\frac{q^\ell-1}{q-1}.
        \]
        
        For each $b\in\F_{q^\ell}^\times$, define
        \[
        W_b:=\{(y_0,\dots,y_{r-1})\in\F_{q^\ell}^r:\ y_{r-1}=b\,y_0^q\}.
        \]
        Equivalently,
        \[
        W_b=\ker(L_b),
        \]
        where
        \[
        L_b:\F_{q^\ell}^r\to\F_{q^\ell},\qquad
        L_b(y_0,\dots,y_{r-1})=y_{r-1}-b\,y_0^q.
        \]
        Since the Frobenius map $y_0\mapsto y_0^q$ is $\F_q$-linear, $L_b$ is $\F_q$-linear. Moreover,
        for every $z\in\F_{q^\ell}$,
        \[
        L_b(0,\dots,0,z)=z,
        \]
        so $L_b$ is surjective. Hence $W_b$ is an $\F_q$-subspace of codimension $\ell$ in
        $\mathbb{V}$, and therefore
        \[
        \dim_{\F_q}(W_b)=(r-1)\ell.
        \]
        
        The following lemma describes exactly which node subspaces are hit by a given $W_b$.
        
        \begin{lemma}\label{lem:nrc-intersection}
        Let $b\in\F_{q^\ell}^\times$.
        \begin{enumerate}[label=(\arabic*)]
        \item For $c\in\F_{q^\ell}^\times$,
        \[
        W_b\cap \mathcal{H}_c\neq\{0\}
        \iff
        c^{r-1}\in b\Sigma.
        \]
        Whenever this holds,
        \[
        \dim_{\F_q}(W_b\cap \mathcal{H}_c)=1.
        \]
        \item One has
        \[
        W_b\cap \mathcal{H}_0=\{0\},
        \qquad
        W_b\cap \mathcal{H}_\infty=\{0\}.
        \]
        \end{enumerate}
        \end{lemma}
        
        \begin{proof}
        A nonzero vector in $\mathcal{H}_c$ has the form
        \[
        x(1,c,c^2,\dots,c^{r-1}),
        \qquad x\in\F_{q^\ell}^\times.
        \]
        Such a vector lies in $W_b$ exactly when
        \[
        x c^{r-1}=b\,x^q,
        \]
        or equivalently,
        \[
        x^{q-1}=b^{-1}c^{r-1}.
        \]
        The image of the map
        \[
        \F_{q^\ell}^\times\to\F_{q^\ell}^\times,\qquad x\mapsto x^{q-1},
        \]
        is precisely $\Sigma$, so this equation has a nonzero solution if and only if
        \[
        b^{-1}c^{r-1}\in\Sigma,
        \]
        that is,
        \[
        c^{r-1}\in b\Sigma.
        \]
        This proves the nontrivial-intersection criterion.
        
        Now assume $W_b\cap\mathcal{H}_c\neq\{0\}$, and choose $x_0\in\F_{q^\ell}^\times$ such that
        \[
        x_0^{q-1}=b^{-1}c^{r-1}.
        \]
        Then all nonzero solutions are precisely the $\F_q^\times$-multiples of $x_0$, so
        \[
        W_b\cap\mathcal{H}_c
        =
        \{\alpha x_0(1,c,\dots,c^{r-1}):\alpha\in\F_q\},
        \]
        which is $1$-dimensional over $\F_q$.
        
        For $c=0$, a vector in $\mathcal{H}_0$ has the form
        \[
        (x,0,\dots,0),
        \]
        and the defining equation of $W_b$ forces $0=b\,x^q$, hence $x=0$. Thus
        \[
        W_b\cap\mathcal{H}_0=\{0\}.
        \]
        Similarly, a vector in $\mathcal{H}_\infty$ has the form
        \[
        (0,\dots,0,x),
        \]
        and the defining equation of $W_b$ forces $x=0$. Hence
        \[
        W_b\cap\mathcal{H}_\infty=\{0\}.
        \]
        \end{proof}
        
        The preceding lemma partitions the nonzero finite parameters $c\in\F_{q^\ell}^\times$
        according to the hit pattern of the repair subspaces. For each $b\in\F_{q^\ell}^\times$, define
        \[
        C_b:=\{c\in\F_{q^\ell}^\times:c^{r-1}\in b\Sigma\}.
        \]
        
        \begin{lemma}\label{lem:block-partition}
        Assume that $(r-1)\mid(q-1)$. Then the distinct nonempty sets $C_b$ form a partition of
        $\F_{q^\ell}^\times$ into
        \[
        \frac{q-1}{r-1}
        \]
        blocks, each of size
        \[
        (r-1)t_\ell(q).
        \]
        \end{lemma}
        
        \begin{proof}
        The quotient group
        \[
        \F_{q^\ell}^\times/\Sigma
        \]
        is cyclic of order
        \[
        \frac{q^\ell-1}{|\Sigma|}=q-1.
        \]
        Consider the power map
        \[
        \psi:\F_{q^\ell}^\times/\Sigma\to \F_{q^\ell}^\times/\Sigma,
        \qquad
        x\Sigma\mapsto x^{r-1}\Sigma.
        \]
        Since $(r-1)\mid(q-1)$ and the quotient is cyclic of order $q-1$, the kernel of $\psi$ has
        size $r-1$. Hence the image has size
        \[
        \frac{q-1}{r-1},
        \]
        and every nonempty fiber has size $r-1$.
        
        Now $C_b$ depends only on the coset $b\Sigma$, and its image in the quotient is precisely
        the fiber of $\psi$ above $b\Sigma$. Therefore the distinct nonempty sets $C_b$ are in
        bijection with the image of $\psi$, so there are
        \[
        \frac{q-1}{r-1}
        \]
        such blocks. Each fiber in the quotient lifts to
        \[
        (r-1)|\Sigma|=(r-1)t_\ell(q)
        \]
        elements of $\F_{q^\ell}^\times$, proving the claim.
        \end{proof}
        
        \begin{theorem}\label{thm:nrc-attainability}
        Assume that $\ell\ge 2$ and $r\ge 2$, and suppose that
        \[
        (r-1)\mid(q-1)
        \qquad\text{and}\qquad
        \frac{q-1}{r-1}\ge 2.
        \]
        Then for every integer $n$ satisfying
        \[
        2(r-1)t_\ell(q)\le n\le q^\ell+1,
        \]
        there exists an $(n,n-r,\ell)$ MDS array code $\mathcal{C}$ over $\F_q$ such that
        \[
        \beta_{\avg}(\mathcal{C})=\beta_{\max}(\mathcal{C})
        =\gamma_{\avg}(\mathcal{C})=\gamma_{\max}(\mathcal{C})
        =\ell(n-1)-(r-1)t_\ell(q).
        \]
        In particular, $\mathcal{C}$ attains the incidence-multiplicity bound simultaneously for
        the average and worst-case repair bandwidth and for the average and worst-case repair I/O.
        \end{theorem}
        
        \begin{proof}
        By Lemma~\ref{lem:block-partition}, we may choose two distinct blocks
        \[
        C_1,\ C_2\subseteq\F_{q^\ell}^\times.
        \]
        Since each block has size
        \[
        (r-1)t_\ell(q),
        \]
        the assumption
        \[
        2(r-1)t_\ell(q)\le n\le q^\ell+1
        \]
        allows us to choose a parameter set
        \[
        \Omega\subseteq \F_{q^\ell}\cup\{\infty\},
        \qquad
        |\Omega|=n,
        \]
        such that
        \[
        C_1\cup C_2\subseteq\Omega.
        \]
        Enumerate
        \[
        \Omega=\{c_1,\dots,c_n\},
        \]
        and set
        \[
        \mathcal{H}_i:=\mathcal{H}_{c_i}
        \qquad (i\in[n]).
        \]
        By Lemma~\ref{lem:nrc-mds}, these subspaces form an $(n,n-r,\ell)$ MDS array-code
        skeleton.
        
        Choose $b_1,b_2\in\F_{q^\ell}^\times$ such that
        \[
        C_1=C_{b_1},
        \qquad
        C_2=C_{b_2},
        \]
        and set
        \[
        W^{(1)}:=W_{b_1},
        \qquad
        W^{(2)}:=W_{b_2}.
        \]
        For each failed node $i\in[n]$, define
        \[
        \widetilde W_i:=
        \begin{cases}
        W^{(2)}, & \text{if } c_i\in C_1,\\[1mm]
        W^{(1)}, & \text{otherwise.}
        \end{cases}
        \]
        Since $C_1\cap C_2=\varnothing$, and since Lemma~\ref{lem:nrc-intersection}(2) shows that
        both $W^{(1)}$ and $W^{(2)}$ avoid $\mathcal{H}_0$ and $\mathcal{H}_\infty$, we have in
        every case
        \[
        \widetilde W_i\cap \mathcal{H}_i=\{0\}.
        \]
        Thus
        \[
        \widetilde W_i\in\mathcal{W}_i
        \qquad\text{for all }i\in[n].
        \]
        
        By construction, $\widetilde W_i$ hits exactly one full block of helper-node subspaces,
        and by Lemma~\ref{lem:nrc-intersection}(1) every nonzero intersection has dimension $1$.
        Therefore
        \[
        \sum_{j\ne i}\dim(\widetilde W_i\cap \mathcal{H}_j)
        =
        (r-1)t_\ell(q)
        \qquad\text{for all }i\in[n].
        \]
        Hence
        \[
        \alpha_i\ge (r-1)t_\ell(q)
        \qquad\text{for all }i\in[n].
        \]
        On the other hand, the incidence-multiplicity bound gives
        \[
        \alpha_i\le (r-1)t_\ell(q).
        \]
        It follows that
        \[
        \alpha_i=(r-1)t_\ell(q),
        \]
        and hence
        \[
        \beta_i(\mathcal{C})=\ell(n-1)-(r-1)t_\ell(q)
        \qquad\text{for all }i\in[n].
        \]
        Consequently,
        \[
        \beta_{\avg}(\mathcal{C})=\beta_{\max}(\mathcal{C})
        =\ell(n-1)-(r-1)t_\ell(q).
        \]
        
        It remains to realize the same value for repair I/O. For each $i\in[n]$, choose a set
        \[
        X_i\subseteq \mathbb{P}(\mathcal{H}_i)
        \]
        of $\ell$ distinct points spanning $\mathcal{H}_i$ as follows:
        \begin{itemize}
        \item if $c_i\in C_1$, require $X_i$ to contain the unique point of
        \[
        \mathbb{P}(W^{(1)}\cap \mathcal{H}_i);
        \]
        \item if $c_i\in C_2$, require $X_i$ to contain the unique point of
        \[
        \mathbb{P}(W^{(2)}\cap \mathcal{H}_i);
        \]
        \item otherwise, choose any $\ell$ distinct spanning points of $\mathbb{P}(\mathcal{H}_i)$.
        \end{itemize}
        By the realization statement in Section~\ref{sec:prelim}, these projective column sets are
        realized by a parity-check realization of an $(n,n-r,\ell)$ MDS array code with node
        subspaces $\mathcal{H}_1,\dots,\mathcal{H}_n$.
        
        For each failed node $i$, the chosen repair subspace $\widetilde W_i$ captures exactly one
        projective column point from each helper node in the corresponding full hit block, and no
        projective column points from the remaining helper nodes. Therefore
        \[
        \lambda_i\ge (r-1)t_\ell(q)
        \qquad\text{for all }i\in[n].
        \]
        Since
        \[
        \gamma_i(\mathcal{C})=\ell(n-1)-\lambda_i,
        \]
        it follows that
        \[
        \gamma_i(\mathcal{C})\le \ell(n-1)-(r-1)t_\ell(q)
        \qquad\text{for all }i\in[n].
        \]
        Combining this with Theorem~\ref{thm:imbound}, we conclude that
        \[
        \gamma_i(\mathcal{C})=\ell(n-1)-(r-1)t_\ell(q)
        \qquad\text{for all }i\in[n].
        \]
        Hence
        \[
        \gamma_{\avg}(\mathcal{C})=\gamma_{\max}(\mathcal{C})
        =\ell(n-1)-(r-1)t_\ell(q).
        \]
        \end{proof}
        
        \begin{remark}
        Relative to the general MDS length bound \(n\le q^\ell+r-1\), the construction above
        covers
        \[
        q^\ell+2-2(r-1)t_\ell(q)
        \]
        integer lengths within this upper-bound range. Thus the proportion of covered lengths is
        \[
        \frac{q^\ell+2-2(r-1)t_\ell(q)}{q^\ell+r-1}.
        \]
        Since
        \[
        t_\ell(q)=\frac{q^\ell-1}{q-1}\approx \frac{q^\ell}{q-1},
        \]
        this proportion is approximately
        \[
        1-\frac{2(r-1)}{q-1}.
        \]
        \end{remark}

        \section{Discussion and Open Problems}

        The incidence-multiplicity bound identifies a sharper obstruction for linear exact repair,
        and Section~\ref{sec:construction} shows that this bound is attained on a broad explicit
        family. Several natural questions remain open.
        
        \begin{enumerate}[label=(\arabic*)]
        
        \item \emph{The attainable length range.}
        The construction of Theorem~\ref{thm:nrc-attainability} attains the
        incidence-multiplicity bound for
        \[
        2(r-1)t_\ell(q)\le n\le q^\ell+1.
        \]
        It is natural to ask whether this interval can be extended. In particular, in the special
        case $r=\ell=2$, results of \cite{liu2026linear} show that attainability can persist
        beyond the range provided by our construction. This suggests the problem of determining
        the largest length range on which the incidence-multiplicity bound can be attained, and in
        particular how far below
        \[
        2(r-1)t_\ell(q)
        \]
        one can go.
        
        \item \emph{Bandwidth versus I/O attainability.}
        Our construction attains the incidence-multiplicity bound simultaneously for repair
        bandwidth and repair I/O. It is natural to ask whether these two notions of sharpness can
        separate. More precisely, does there exist a code length for which the incidence-multiplicity
        bound is attained for repair bandwidth, but not for repair I/O?
        
        \item \emph{The short-length regime.}
        Remark~\ref{rem:multiplicity-equality} shows that equality in the incidence-multiplicity bound can occur
        only if
        \[
        n\ge 1+(r-1)\frac{q^\ell-1}{q-1}.
        \]
        This leaves open the regime
        \[
        n<1+(r-1)\frac{q^\ell-1}{q-1},
        \]
        where the incidence-multiplicity bound cannot be sharp. In that range, it is natural to
        ask what other mechanisms govern the optimal lower bounds for repair bandwidth and repair
        I/O.
        
        \item \emph{The nondivisibility regime.}
        The construction in Section~\ref{sec:construction} requires
        \[
        r-1\mid(q-1).
        \]
        It is therefore natural to ask whether the incidence-multiplicity bound can still be
        attained when
        \[
        r-1\nmid(q-1).
        \]
        \end{enumerate}

\phantomsection
\addcontentsline{toc}{section}{\refname}
\bibliographystyle{alphaurl}
\bibliography{references}


\end{document}